\documentclass[journal,comsoc]{IEEEtran}%

\usepackage{amsmath,amsthm,amssymb}
\usepackage{mathtools}
\usepackage{amsfonts}

\usepackage{graphicx}
\usepackage{subfigure}
\usepackage{float}

\usepackage{epstopdf}

\usepackage{tabularx}

\usepackage{hyperref}

\usepackage{cuted}
\usepackage{multicol}

\usepackage{stfloats}

\graphicspath{{Figures/}{Codes/}}

\newtheorem{thm}{Theorem}
\newtheorem{cor}{Corollary}
\newtheorem{prop}[thm]{Proposition}
\newtheorem{defn}[thm]{Definition}

\newtheorem{rem}{Remark}

\begin{document}


\title{Swarm Antenna Arrays: From Deterministic to Stochastic Modeling
\\ 
\thanks{This work is supported by the National Natural Science Foundation of China under Grant 12141107, the Key Research and Development Program of Wuhan under Grant 2024050702030100, and the Interdisciplinary Research Program of HUST (2023JCYJ012).}
\thanks{T. Mi, M. Feng, R. Shao and R. Qiu are with the School of Electronic Information and Communications, Huazhong University of Science and Technology, Wuhan 430074, China (e-mail: mitiebin@hust.edu.cn; feng\_miyu@hust.edu.cn; shaoruichu@hust.edu.cn; caiming@hust.edu.cn).}
\thanks{C. Zeng is with the National Key Laboratory of Radar Signal Processing, Xidian University, Xi'An 710071, China. (e-mail: czeng@mail.xidian.edu.cn)}
}

\author{\IEEEauthorblockN{Tiebin Mi,~\IEEEmembership{Member,~IEEE,} 
Miyu Feng,~\IEEEmembership{Student Member,~IEEE,}
Ruichu Shao,~\IEEEmembership{Student Member,~IEEE,} \\
Cao Zeng,~\IEEEmembership{Member,~IEEE,} 
and Robert Caiming Qiu,~\IEEEmembership{Fellow,~IEEE}} \\
}

\maketitle

\begin{abstract}
Swarm antenna arrays, composed of spatially distributed antennas mounted on unmanned agents, offer unprecedented flexibility and adaptability for wireless sensing and communication. However, their reconfigurable architecture, susceptibility to collisions, and inherently stochastic nature present significant challenges to realizing collaborative gain. It remains unclear how spatial coordination, positional perturbations, and large-scale topological configurations affect coherent signal aggregation and overall system performance. This paper investigates the feasibility of achieving coherent beamforming in such systems from both deterministic and stochastic perspectives. First, we develop a rigorous theoretical framework that characterizes the necessary and sufficient conditions for the emergence of grating lobes in multiple linear  configurations. Notably, we show that for dual linear arrays, the classical half-wavelength spacing constraint can be safely relaxed without introducing spatial aliasing. This result challenges traditional array design principles and enables more flexible, collision-aware topologies. Second, we present a theoretical analysis, supported by empirical validation, demonstrating that coherent gain can be approximately preserved under realistic positional perturbations. Our results reveal that spatial perturbations introduce measurable degradation in the main lobe, an effect that cannot be mitigated merely by increasing the number of antennas. Instead, the primary benefit of scaling lies in reducing the variance of perturbation-induced fluctuations. Finally, we examine the emergent deterministic behavior of large-scale disordered arrays by analyzing the spectral properties of the associated Euclidean random matrices. Despite positional randomness, the spectral distributions are shown to asymptotically converge to deterministic limits, thereby revealing inherent structural regularities within disordered configurations. Together, these findings offer new theoretical foundations and practical design insights for enabling advanced functionalities in swarm antenna arrays.
\end{abstract}


\section{Introduction}\label{S:Introduction}

\IEEEPARstart{U}{nmanned} agents have emerged as transformative enablers for target detection and wireless communication due to their inherent flexibility, mobility, and adaptability \cite{zeng2016wireless, mozaffari2019tutorial}. Recent advancements highlight the capabilities of unmanned aerial vehicles (UAVs) as versatile platforms for wireless sensing and communication tasks \cite{mu2023uav, meng2023uav}. By exploiting their mobility, Unmanned agents can dynamically establish favorable line-of-sight (LoS) links \cite{sabzehali20213d}, thereby enabling efficient data relaying \cite{zeng2016throughput, zhan2011wireless}, channel probing \cite{khuwaja2018survey, mao2024survey}, and wide-area coverage \cite{li2019near}. In sensing applications, UAVs have been employed in various tasks such as source localization \cite{li2021multiple, zhang2022efficient} and radio environment mapping \cite{liu2023uav, zeng2021simultaneous}, often using techniques such as received signal strength fingerprinting \cite{soltani2020rf}, time-of-arrival \cite{zhang2025cooperative}, and angle-of-arrival \cite{jiang2020novel} estimation. On the communication front, UAV-mounted base stations have been extensively studied for providing temporary wireless coverage in disaster recovery and rural connectivity scenarios \cite{zhang20213d, lyu2016placement}.

The cost-effectiveness and operational simplicity of unmanned agents facilitate their deployment in swarm configurations \cite{mozaffari2017wireless, yuan2018ultra}. In the context of wireless sensing and communication, a particularly promising paradigm is the swarm antenna array, in which each unmanned agent carries a dedicated antenna element, collectively forming a large-scale, distributed, and coordinated network \cite{namin2012analysis, mozaffari2018communications, diao2022unmanned, diao2024experimental, zhang2024uav}. Swarm antenna arrays represent a fundamentally new class of reconfigurable architectures and extend the capabilities of multiple-input multiple-output (MIMO) systems.

Two primary methodological approaches can be employed to investigate the collaborative gain enabled by swarm antenna arrays, depending on the specific task and operational context. In communication-oriented scenarios, swarm-based wireless systems are typically regarded as extensions of classical MIMO architectures (e.g., \cite{mozaffari2017wireless, mozaffari2018communications, zhang2024uav}). This perspective leverages well-established tools from wireless communications to analyze system capacity, perform channel estimation, and evaluate overall network performance. The second approach adopts a physical-layer-centric view, treating the swarm antenna array as a tightly coupled system and focusing on intra-array coherent gain at the wavelength scale (e.g., \cite{namin2012analysis, diao2022unmanned, tuzi2023satellite, diao2024experimental}). Analytical methods from electromagnetic theory, array signal processing, and statistical modeling are frequently employed to examine how the swarm topology influences the collective system behavior.

Compared to the communication-oriented approach, the physics-based methodology adopts a lower-level, first-principles perspective grounded in fundamental theory \cite{rocca2016unconventional, nanzer2021distributed}. It directly analyzes how spatial topology influences one of the most fundamental performance metrics---the radiation pattern of the entire antenna array---thereby facilitating a deeper understanding of the deterministic and statistical behaviors that arise from large-scale configurations.

In this paper, we investigate the feasibility of achieving coherent beamforming using swarm antenna arrays. In conventional array processing theory, a well-established and fundamental principle for uniform linear or planar arrays is the half-wavelength spacing criterion, which serves as a stringent guideline to suppress the emergence of grating lobes and to mitigate angular ambiguities across the entire field of view \cite{van2002optimum}. Considering the risk of inter-agent collisions in practical agents deployments, a fundamental yet unexplored question arises: can the half-wavelength spacing constraint be relaxed without compromising beamforming performance?

A second critical issue involves positional perturbations of the antenna elements. Unlike traditional arrays with fixed and precisely calibrated topologies, swarm antenna arrays are deployed on spatially distributed, independently controlled agents. The array topology is inherently uncertain due to environmental disturbances and mobility-induced deviations. These perturbations degrade coherent signal superposition and introduce random fluctuations into the radiation pattern. Consequently, understanding the impact of such positional uncertainties is essential to ensure the performance of swarm-based antenna systems.

Moreover, in many scenarios, agents may be randomly or irregularly positioned, resulting in a disordered configuration. Notably, when these positions are distributed randomly within a known region, certain aggregate measurements exhibit statistical regularities that asymptotically converge to deterministic patterns. This phenomenon reflects a fundamental principle in large-scale systems, where microscopic randomness leads to emergent macroscopic order.

\subsection{Contribution}
The primary contributions of this work are summarized as follows:

\begin{itemize}
\item \textbf{Necessary and sufficient conditions for the existence of grating lobes.} We establish a rigorous theoretical framework that characterizes the necessary and sufficient conditions for the emergence of grating lobes in multiple linear antenna arrays. For this class of topologies, we derive explicit constraints on both the inter-element spacing within each linear sub-array and the relative positioning of their leading elements. A significant and surprising finding is that, for dual linear arrays, the classical half-wavelength spacing criterion can be relaxed without introducing spatial aliasing. This result fundamentally challenges traditional array design principles and paves the way for more flexible and collision-aware swarm antenna architectures.

\item \textbf{Statistical analysis of fluctuations induced by spatial perturbations.} We present a theoretical analysis, supported by empirical validation, to demonstrate that coherent gain can be approximately maintained under realistic models of positional perturbations. The analysis provides an explicit characterization of how the induced fluctuations depend on key system parameters. It is shown that spatial perturbations result in a measurable degradation of the main lobe, an effect that cannot be mitigated simply by increasing the number of antenna elements. The primary benefit of scaling up the array lies in suppressing the variance of the perturbation-induced fluctuations.

\item \textbf{Emergent deterministic behavior in large disordered antenna arrays.} 

We investigate the statistical regularities that asymptotically converge to deterministic patterns in large-scale disordered swarm antenna arrays. In particular, we analyze the Euclidean random matrix associated with such arrays and show that, when antenna positions are randomly distributed within a known region, the spectral properties exhibit deterministic behavior in the large-system limit. This emergent order, arising from spatial randomness, reveals fundamental structural insights that can be leveraged for future applications.

\end{itemize}

\subsection{Outline}

The remainder of this paper is organized as follows. Section~\ref{S:II} introduces the wave equation and Green's function, which provide the foundation for the modeling approach. Section~\ref{S:III} presents a deterministic framework for coherent beamforming and characterizes the conditions for grating lobes in multiple linear topologies. Section~\ref{S:IV} analyzes the impact of spatial perturbations on array performance. Section~\ref{S:V} investigates the emergent deterministic behavior in disordered arrays through spectral analysis of Euclidean random matrices. Finally, Section~\ref{S:VI} concludes the paper.

\section{Wave equation and Green's function}\label{S:II}

We initiate our exploration by considering an idealized point source emitting radiation radially, uniformly in all directions.  In this scenario, the three-dimensional differential wave equation governing the behavior of the wave function $\psi (x, y, z; t)$ is simplified to a one-dimensional form \cite{hecht2012optics}
\begin{equation}\label{E:OneDimensionalWaveEquation}
\frac{ \partial ^2 } { \partial r^2} ( r \psi ) = \frac{1}{c^2} \frac{ \partial ^2 } { \partial t^2} ( r \psi ) . 
\end{equation}
Here, with a slight abuse of notations, we use $\psi (r; t)$ to denote a wave function of both the spatial variable $r$ (distance to the source) and the time variable $t$. The solution to \eqref{E:OneDimensionalWaveEquation} is given by
\[
  \psi (r; t) = \frac{f ( c t - r )}{r} .
\]
This solution represents a spherical wave progressing radially outward from the source at a constant speed $c$, and exhibiting an arbitrary functional form $f$. For harmonic spherical waves, the solution takes the form
\begin{equation}\label{E:HarmonicSphericalWaves}
  \psi (r; t) = \frac{E_0}{r} e^{j ( \Omega t - 2 \pi r / \lambda )} = \frac{ e^{ - j 2 \pi r / \lambda } } { r } E_0 e^{ j \Omega t }.
\end{equation}
Here, $E_0$ denotes the amplitude of the source.

From \eqref{E:HarmonicSphericalWaves}, we observe that a harmonic spherical wave consists of two components: a temporal factor $e^{ j \Omega t }$, and a spatial attenuation term $\frac{ e^{ - j 2 \pi r / \lambda } } { r }$, which closely resembles the free-space Green's function of the Helmholtz equation (up to a constant factor). It is well known that the fundamental solution, or Green's function, to the inhomogeneous Helmholtz equation in free space is given by
\begin{equation}\label{E:GreenFunction}
G(\mathbf{r}, \mathbf{r}') = \frac{ e^{ - j 2 \pi \| \mathbf{r} - \mathbf{r}' \| / \lambda } }{ 4 \pi \| \mathbf{r} - \mathbf{r}' \| },
\end{equation}
where $\| \mathbf{r} - \mathbf{r}' \|$ denotes the Euclidean distance between the observation point $\mathbf{r}$ and the source location $\mathbf{r}'$. This Green's function characterizes the field response at an arbitrary point in space due to a unit-amplitude point source.

\section{Deterministic Modeling for Coherent Beamforming of Swarm Antenna Arrays}\label{S:III}

We investigate the coherent beamforming capabilities of swarm antenna arrays, where each antenna element coordinates its transmission or reception to achieve constructive interference in a desired direction. The theoretical foundations underlying this behavior remain largely unexplored. In this paper, we focus on a simplified scenario in which both the antenna elements and the target are confined to a common plane, as illustrated in Fig.~\ref{F:Plane_MultipleLinear}. This simplification preserves the essential principles---much like a uniform linear array can be regarded as the projection of a uniform planar array onto a single dimension.

\begin{figure}[htbp]
  \centering
  \includegraphics[width=\columnwidth]{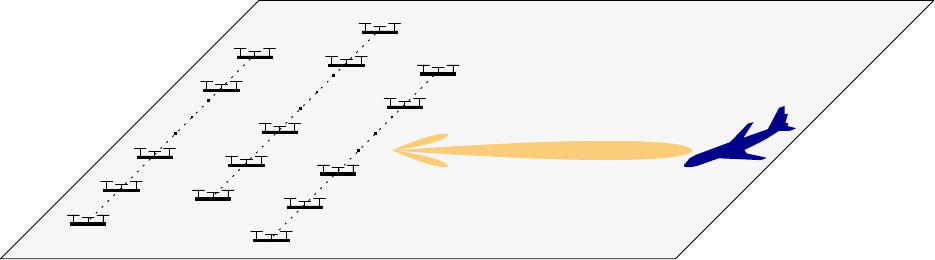}
  \caption{Coherent beamforming via swarm antenna arrays.}
  \label{F:Plane_MultipleLinear}  
\end{figure}

\subsection{Coordinated Beamforming in the Far-Field Regime}

In the far-filed regime, the distances between the antenna elements and the target exhibit a near-linear relationship with the reference distance. According to the geometry illustrated in Fig.~\ref{F:DistanceApproximation}, the distance $r_n$ is expressed as
\begin{equation}\label{E:r_n}
  \begin{aligned}
  r_n = & \sqrt { \left( r \sin \theta - x_n \right)^2 + \left( r \cos \theta - y_n \right)^2 } \\
  = & \sqrt { r^2 - 2 r x_n \sin \theta - 2 r y_n \cos \theta + x_n^2 + y_n^2  } \\
  = & r \sqrt { 1 - 2 \frac{ x_n \sin \theta + y_n \cos \theta }{ r } + \frac{ x_n^2 + y_n^2 }{ r^2 } } .
  \end{aligned}
\end{equation}
Under the far-field approximation, where $ x_n \sin \theta + y_n \cos \theta \ll r $, or equivalently $ (x_n \sin \theta + y_n \cos \theta) / r \ll 1$. Applying a first-order Taylor expansion $\sqrt{1+x} \approx 1+x/2$ for small $x$, and neglecting higher-order terms, the distance $r_n$ can be approximated as
\begin{equation}\label{E:TaylorExpansion}
  \begin{aligned}
  & \sqrt { 1 - 2 \frac{ x_n \sin \theta + y_n \cos \theta }{ r } + \frac{ x_n^2 + y_n^2 }{ r^2 } } \\
\approx & 1 - \frac{ x_n \sin \theta + y_n \cos \theta }{ r } + \frac{ x_n^2 + y_n^2 }{ 2 r^2 } 
\approx 1 - \frac{ x_n \sin \theta + y_n \cos \theta }{ r }.
  \end{aligned}
\end{equation}
Substituting \eqref{E:TaylorExpansion} into \eqref{E:r_n}, we obtain the linearized expression
\[
r_n \approx r - ( x_n \sin \theta + y_n \cos \theta ) .
\]

\begin{figure}[htbp]
  \centering
  \includegraphics[width=0.5\columnwidth]{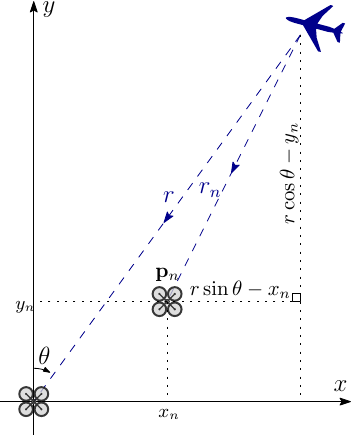}
  \caption{Near-linear distance approximation in the far-field regime.}
  \label{F:DistanceApproximation}  
\end{figure}

For coherent beamforming in the far-field regime, the attenuation represented by the denominators in \eqref{E:HarmonicSphericalWaves} and \eqref{E:GreenFunction} across different elements can be considered equivalent. The primary contribution to the array response arises from the exponential term, which captures the phase differences among the antenna elements. Let $w_n$ denote the weight associated with the $n$-th antenna element. The response in the direction $\theta$ is then given by
\begin{equation}\label{E:BeamPatternGeneral}
  f_{\mathbf{w}}(\theta) = \frac{ 1 }{ \sum_{n=1}^{N} | w_{n} | } \sum_{n=1}^{N} w_{n} e^{j 2 \pi [ x_{n} \sin \theta + y_{n} \cos \theta ] / \lambda} .
\end{equation}
Here, the subscript indicates the beam pattern function associated with the weight vector $\mathbf{w}$. Notably, the normalization factor $1/ \sum_{n=1}^{N} |w_n|$ ensures that the overall array response maintains a consistent amplitude, facilitating comparisons across various scenarios. A direct consequence of this normalization is that $|f_{\mathbf{w}}(\theta)| \leq 1$ for all $\theta$.

If the weights are chosen as
\begin{equation}\label{E:WeightSelection0}
w_n = |w_n| e^{-j 2 \pi [ x_n \sin \theta_s + y_n \cos \theta_s ] / \lambda} ,
\end{equation}
then the array response satisfies $f_{\mathbf{w}}(\theta_s) = 1$, indicating that the beam achieves its maximum value at the desired steering direction $\theta_s$. The underlying principle is that this weight selection effectively compensates for the phase differences induced by the spatial distribution of the antenna elements relative to $\theta_s$, resulting in coherent signal addition along the desired direction. A widely adopted choice is to set $| w_n | =1$ and 
\begin{equation}\label{E:WeightSelection1}
w_n = e^{-j 2 \pi [ x_n \sin \theta_s + y_n \cos \theta_s ] / \lambda} ,
\end{equation}
corresponding to uniform excitation across all antenna elements. Although alternative weight selection strategies exist, such as optimization-based methods and adaptive algorithms, we adopt phase-difference compensation in the following analysis due to its simplicity and analytical tractability.

\subsection{Multiple Linear Topologies and Grating Lobe Conditions}

As a generalization of traditional linear arrays, the multiple-linear topology represents the simplest and most analytically tractable configuration within swarm antenna arrays. In this topology, antenna elements are arranged along several distinct linear formations. Despite its relative simplicity, the multiple linear arrangement captures essential characteristics of beam steering capability, making it a valuable foundation for exploring more complex swarm configurations.

\begin{figure}[htbp]
  \centering
  \includegraphics[width=0.75\columnwidth]{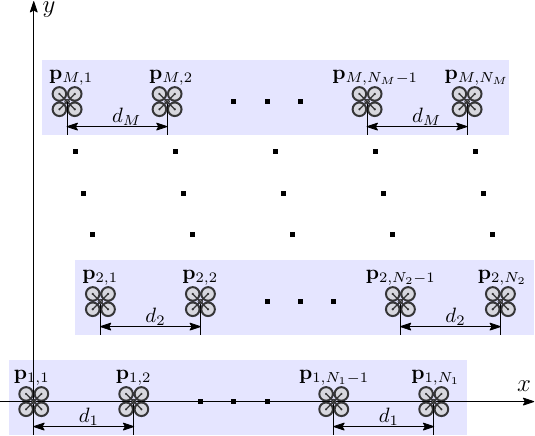}
  \caption{Multiple linear topology. Antenna elements are arranged along several distinct linear arrays.}
  \label{F:MultipleLinearTopology}  
\end{figure}

In uniform linear arrays, a fundamental design constraint is the suppression of grating lobes---undesired secondary maxima in the beam pattern that result from the periodic spatial sampling of the wavefront. These lobes can significantly degrade directional performance by introducing ambiguities and reducing energy concentration in the main beam. To avoid grating lobes, a classical design criterion requires that the inter-element spacing be less than or equal to half the wavelength ($d \leq \lambda/2$). We examine the conditions governing the existence of grating lobes in multiple-linear array configurations.

Mathematically, the principle of grating lobes is defined based on the periodicity of the array response, as formalized below.

\begin{defn}[Principle of Grating Lobes]
Let $f_{\mathbf{w}} (\cdot)$ denote the beam pattern function, and let $[\theta_1, \theta_2]$ be a non-empty angular interval. If there exists a mapping $\mathfrak{M} (\cdot)$ such that for every $\mathbf{w}$ and all $\theta \in (\theta_1, \theta_2)$, $\mathfrak{M} (\theta) \neq \theta $ and 
\[
f_{\mathbf{w}} (\theta) = f_{\mathbf{w}}( \mathfrak{M} (\theta) ) ,
\]
then $f_{\mathbf{w}}(\cdot)$ is referred to as a periodic beam pattern and $\mathfrak{M} (\cdot)$ is called the period mapping. If the absolute radiation maxima occur within $[\theta_1, \theta_2]$, then the pattern exhibits grating lobes. 
\end{defn}

Consider the topology illustrated in Fig.~\ref{F:MultipleLinearTopology}. Under this configuration, the general array response in \eqref{E:BeamPatternGeneral} simplifies to
\begin{equation}\label{E:BeamPattern}
\begin{aligned}
& f_{\mathbf{w}}(\theta) \\
= & \frac{ 1 }{ \sum_{n=1}^{N_1} | w_{1,n} | + \cdots + \sum_{n=1}^{N_M} | w_{M,n} | } \Bigl( \sum_{n=1}^{N_1} w_{1,n} e^{j 2 \pi (n-1) d_1 \sin \theta / \lambda} \\
& + e^{j 2 \pi [ x_{2,1} \sin \theta + y_{2,1} \cos \theta ] / \lambda} \sum_{n=1}^{N_2} w_{2,n} e^{j 2 \pi (n-1) d_2 \sin \theta / \lambda}  \\
& + \cdots \\
& + e^{j 2 \pi [ x_{M,1} \sin \theta + y_{M,1} \cos \theta ] / \lambda} \sum_{n=1}^{N_M} w_{M,n} e^{j 2 \pi (n-1) d_M \sin \theta / \lambda} \Bigr) .
\end{aligned}
\end{equation}
A distinctive feature of this function is that it explicitly depends on the positioning of the first element of each sub-array relative to a global coordinate system. In particular, the terms associated with each individual linear array correspond to harmonic (complex exponential) functions of $d_m \sin \theta / \lambda$, reflecting the regular spacing of elements along each line.

\begin{thm}\label{T:theorem1}
A necessary and sufficient condition for the periodicity at $\theta$ of the beam pattern $f_{\mathbf{w}} ( \cdot )$ associated with multiple linear  topology is that for $i=1, \ldots, M$,
\begin{equation*}\label{E:Condition_1}
(C1) \qquad \frac{d_i \left( \sin \theta - \sin (\mathfrak{M} (\theta)) \right) }{\lambda} \in \mathbb{Z} \setminus \{0\},
\end{equation*}
and for $l = 2, \ldots, M$, 
\begin{equation*}\label{E:Condition_2}
(C2) \quad\frac{ x_{l,1} \left( \sin \theta - \sin ( \mathfrak{M} (\theta) ) \right) + y_{l,1} \left( \cos \theta - \cos ( \mathfrak{M} (\theta) ) \right) }{ \lambda } \in \mathbb{Z} .
\end{equation*}
\end{thm}

\begin{proof}
We begin by proving the necessary condition. The response function $f_{\mathbf{w}} (\theta)$ is expressed as a sum of complex exponentials with coefficients determined by $\mathbf{w}$ (see \eqref{E:BeamPattern}). Suppose there exists $\mathfrak{M} (\theta) \neq \theta$ such that 
$f_{\mathbf{w}}(\theta) = f_{\mathbf{w}}( \mathfrak{M} (\theta) )$ for any choice of $\mathbf{w}$. Then each corresponding exponential term in the sum must be equal.

That is, for $n_1 = 1, \ldots, N_1$,
\begin{equation}\label{E:exp_n1}
  e^{ j 2 \pi (n_1 - 1) d_1 \sin \theta / \lambda } = e^{ j 2 \pi (n_1 -1) d_1 \sin ( \mathfrak{M} (\theta) ) / \lambda } \\
\end{equation}
and for $l = 2, \ldots, M$, and $n_l = 1, \ldots, N_l$,

\begin{equation}\label{E:exp_exp_nl}
\begin{aligned}
& e^{ j 2 \pi ( x_{l,1} \sin \theta + y_{l,1} \cos \theta ) / \lambda} e^{ j 2 \pi (n_l-1) d_l \sin \theta / \lambda } \\ 
= & e^{ j 2 \pi ( x_{l,1} \sin ( \mathfrak{M} (\theta) ) + y_{l,1} \cos ( \mathfrak{M} (\theta) ) / \lambda} e^{ j 2 \pi (n_l-1) d_l \sin ( \mathfrak{M} (\theta) ) / \lambda } .
\end{aligned}
\end{equation}
When $n_l = 1$, the exponential factors associated with $n_l-1$ in \eqref{E:exp_exp_nl} are all reduced to 1, so that for $l = 2, \ldots, M$, we obtain
\begin{equation}\label{E:exp_x_y}
e^{j 2 \pi ( x_{l,1} \sin \theta + y_{l,1} \cos \theta ) / \lambda} = e^{j 2 \pi ( x_{l,1} \sin ( \mathfrak{M} (\theta) ) + y_{l,1} \cos ( \mathfrak{M} (\theta) ) / \lambda} ,
\end{equation}
and, consequently, for $n_l = 2, \ldots, N_l$,
\begin{equation}\label{E:exp_nl}
e^{j 2 \pi (n_l-1) d_l \sin \theta / \lambda} = e^{j 2 \pi (n_l-1) d_l \sin ( \mathfrak{M} (\theta) ) / \lambda} .
\end{equation}

From \eqref{E:exp_n1} and \eqref{E:exp_nl}, we deduce that for each $i=1, \ldots, M$,
\[
 \frac{d_i \left( \sin \theta - \sin ( \mathfrak{M} (\theta) ) \right) }{\lambda} \in \mathbb{Z} \setminus \{0\}.
\]
Similarly, from \eqref{E:exp_x_y}, we have for $l = 2, \ldots, M$, 
\[
\frac{ x_{l,1} \left( \sin \theta - \sin ( \mathfrak{M} (\theta) ) \right) + y_{l,1} \left( \cos \theta - \cos ( \mathfrak{M} (\theta) ) \right) }{ \lambda } \in \mathbb{Z} .
\]

The sufficient condition is straightforward. If conditions (C1) and (C2) hold, the complex exponential terms in \eqref{E:BeamPattern} for $\theta$ and $\mathfrak{M} (\theta)$ are equal due to their inherent periodicity. Consequently, 
\[
f_{\mathbf{w}}(\theta) = f_{\mathbf{w}}( \mathfrak{M} (\theta) )
\]
for every weight vector $\mathbf{w}$.
\end{proof}

\begin{rem}
  Condition (C1) specifies the requirements on the inter-element distances within each linear array, while (C2) governs the positioning of the first element of each sub-array. For grating lobes to exist, both (C1) and (C2) must be satisfied simultaneously. This simultaneous satisfaction imposes strict structural constraints on both the spacings and the placements of the sub-arrays, emphasizing the inherently coupled nature of array geometry.
\end{rem}

\begin{rem}
  Condition (C1) further implies that a necessary condition for the existence of grating lobes is that the ratio between any two inter-element distances, $d_i/d_j$, must be a rational number.  This observation indicates that irrational ratios between inter-element spacings inherently prevent the alignment of phase fronts necessary for coherent secondary maxima, thus naturally suppressing grating lobe formation.
\end{rem}

\subsection{Dual Linear Topologies}

To gain a more informative understanding of Theorem~\ref{T:theorem1}, we focus on the dual linear topology where the inter-element spacings are identical, i.e., $d_1 = d_2 = d$, as illustrated in Fig.~\ref{F:DualLinear}.  In this scenario, Conditions (C1) and (C2) simplify to 
\begin{equation}\label{E:Condition_11}
d/\lambda \left( \sin \theta - \sin (\mathfrak{M} (\theta)) \right) = p \in \mathbb{Z} \setminus \{ 0 \} ,
\end{equation}
and
\begin{equation}\label{E:Condition_21}
x_{2,1}/\lambda \left( \sin \theta - \sin ( \mathfrak{M} (\theta) ) \right) + y_{2,1} /\lambda \left( \cos \theta - \cos ( \mathfrak{M} (\theta) ) \right) = q \in \mathbb{Z} .
\end{equation}

\begin{figure}[htbp]
  \centering
  \includegraphics[width=0.75\columnwidth]{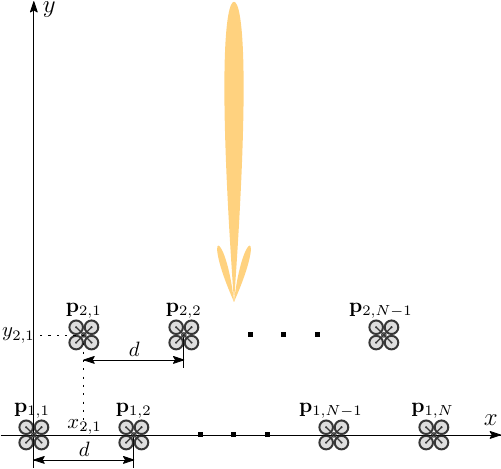}
  \caption{Dual linear topology. The inter-element spacings within both sub-arrays are identical, i.e., $d_1 = d_2 = d$.}
  \label{F:DualLinear}
\end{figure}

From \eqref{E:Condition_11}, we obtain
\[
  \sin \theta - \sin ( \mathfrak{M} (\theta) ) = \frac{p}{d / \lambda} .
\]
Similarly, from \eqref{E:Condition_21}, we have
\[
  \cos \theta - \cos ( \mathfrak{M} (\theta) ) = \frac{q d / \lambda - p x_{2,1} / \lambda}{ (d / \lambda) (y_{2,1} / \lambda) } .
\]
Combining these two expressions, it follows that
\begin{equation}\label{E:SufficientCondition_1}
 \left( \frac{p}{d / \lambda} \right)^2 + \left( \frac{q d / \lambda - p x_{2,1} / \lambda}{ (d / \lambda) (y_{2,1} / \lambda)} \right)^2 = 2 - 2 \cos( \theta - \mathfrak{M} (\theta) ) .
\end{equation}
Thus, \eqref{E:SufficientCondition_1} establishes a sufficient condition under which the beam pattern loses its periodicity.

\begin{thm}\label{T:SufficientGratingLobes}
For a dual linear swarm antenna array, if for all $p \in \mathbb{Z} \setminus \{0\}$ and $q \in \mathbb{Z}$ the inequality
\begin{equation*}
(C3) \qquad  \left( \frac{p}{d / \lambda} \right)^2 + \left( \frac{q d / \lambda - p x_{2,1} / \lambda}{ (d / \lambda) (y_{2,1} / \lambda)} \right)^2 > 4 
\end{equation*}
holds, then the beam pattern is not periodic and no grating lobes exist.
\end{thm}

\begin{proof}
Assume Condition (C3) holds. Then, for every $\theta \in [0, 2 \pi)$, no mapping $\mathfrak{M} (\theta)$ can satisfy the identity given by \eqref{E:SufficientCondition_1}. In other words, there exists no mapping $\mathfrak{M} (\theta)$ that simultaneously satisfies \eqref{E:Condition_11} and \eqref{E:Condition_21}. Therefore, the beam pattern cannot be periodic, and grating lobes do not exist.
\end{proof}

For conventional uniform linear arrays, a fundamental and well-established design rule states that grating lobes will occur if the element spacing exceeds half a wavelength. This design principle is embedded within Condition (C3). Specifically, from the first square term in (C3), it follows that if $d < \lambda/2$, then (C3) is automatically satisfied, irrespective of the coordinates $(x_{2,1}, y_{2,1})$.

A particularly important and surprising finding from Theorem~\ref{T:SufficientGratingLobes} is that, for dual linear arrays, the classical half-wavelength limit is no longer the strict threshold governing grating lobe formation. Instead, Condition (C3) reveals a more relaxed spacing requirement for the suppression of grating lobes compared to the conventional half-wavelength constraint. This result significantly broadens the feasible design space for swarm antenna arrays, providing greater flexibility in array topology planning and deployment.

\begin{figure}[!htbp]
  \centering
  \subfigure[]{
  \label{Fig:GreaterHalfWavelength:0_8}
  \includegraphics[height=0.25\columnwidth]{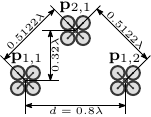}}
  \hspace{6mm}
  \subfigure[]{
  \label{Fig:GreaterHalfWavelength:0_5774}
  \includegraphics[height=0.3\columnwidth]{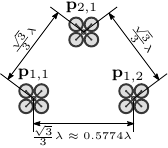}}
  \caption{Relaxation of the classical half-wavelength constraint.}
  \label{Fig:GreaterHalfWavelength}
\end{figure}

To demonstrate the breaking of the half-wavelength limit, we consider the case where $d = 0.8\lambda$ and $x_{2,1} = 0.4\lambda$. In this setting, if
\[
y_{2,1} < \frac{1}{\sqrt{9.75}} \approx 0.3203,
\]
then (C3) holds for all $p \in \mathbb{Z} \setminus \{0\}$ and $q \in \mathbb{Z}$. This verification is straightforward. For simplicity, we set $y_{2,1} = 0.32 \lambda$. Notably, the distance between $\mathbf{P}_{2,1}$ and $\mathbf{P}_{1,1}$ is approximately $0.5122 \lambda$, which exceeds half a wavelength. The topology of the leading elements of each line is illustrated in Fig.~\ref{Fig:GreaterHalfWavelength:0_8}.

To validate the absence of grating lobes, we conduct a numerical experiment in which the first line consists of 50 antenna elements and the second line consists of 49 antenna elements. The weights are selected according to \eqref{E:WeightSelection1}, and the steering angle $\theta_s$ is swept from $-\pi/2$ to $\pi/2$. The corresponding steered beam patterns are shown in Fig.~\ref{Fig:Pattern0}. It can be observed that no grating lobes appear across the entire steering range. Although higher side lobes are present for certain steering angles, these side lobes do not qualify as grating lobes.

\begin{figure}[!htbp]
  \centering
  \subfigure[]{
  \label{Fig:Pattern0:Mesh}
  \includegraphics[width=0.48\columnwidth]{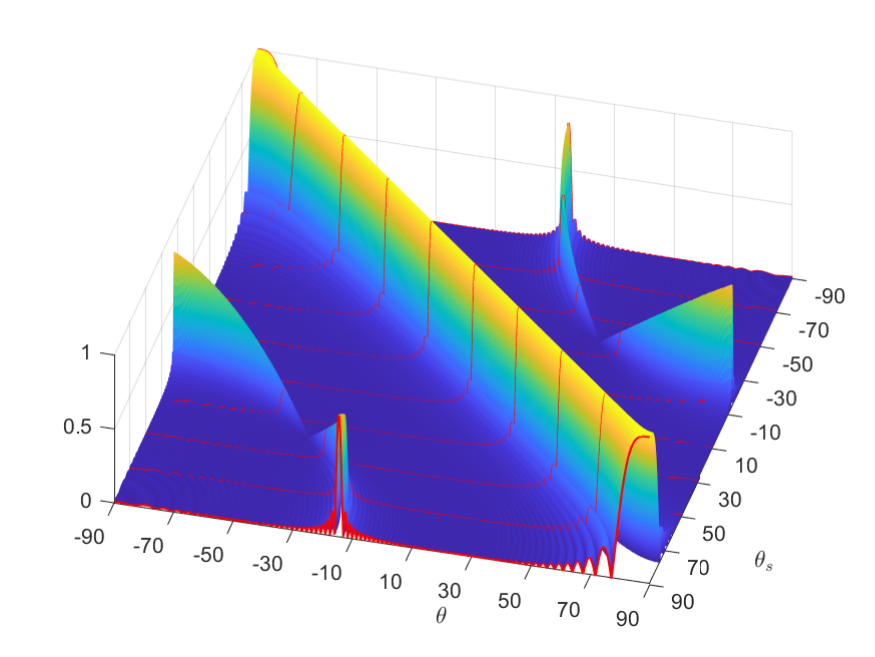}}
  \subfigure[]{
  \label{Fig:Pattern0:Plot}
  \includegraphics[width=0.48\columnwidth]{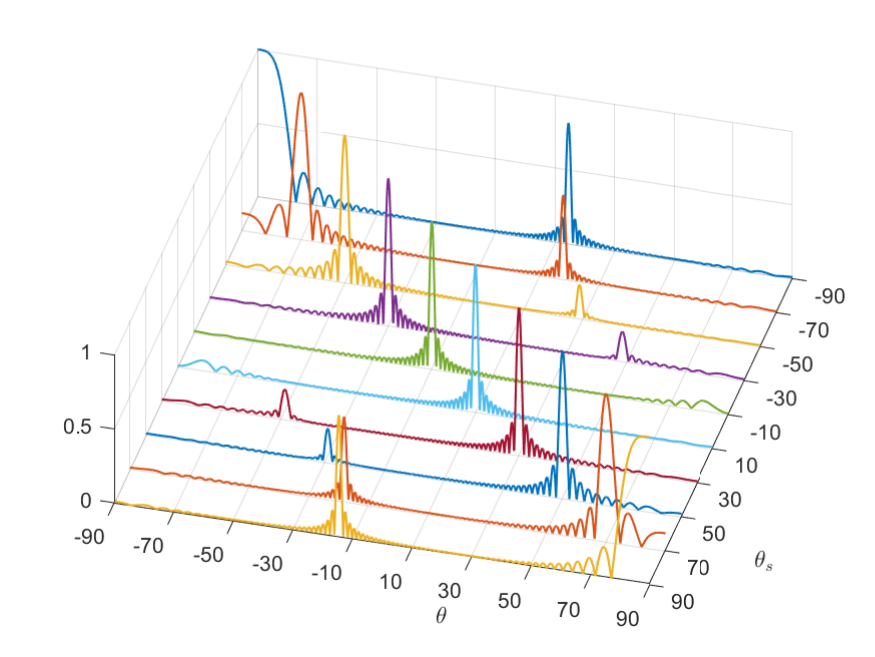}}
  \caption{Beam patterns associated with the dual linear topology where $d=0.8\lambda$, $x_{2,1}=0.4 \lambda$ and $y_{2,1} = 0.32\lambda$. The steering angle $\theta_s$ is swept from $-\pi/2$ to $\pi/2$. No grating lobes are observed across the entire range.}
  \label{Fig:Pattern0}
\end{figure}

We now restrict the topology of the leading elements to a regular triangular configuration. Consider a scenario where $d = \sqrt{3} \lambda / {3} \approx 0.5774 \lambda$, $x_{2,1} = \sqrt{3} \lambda / {6} \approx 0.2887 \lambda$ and $y_{2,1} = 0.5 \lambda$, as illustrated in Fig.~\ref{Fig:GreaterHalfWavelength:0_5774}. In this configuration, for all $p \in \mathbb{Z} \setminus \{0\}$ and $q \in \mathbb{Z}$, the following identity holds
\begin{equation*}
  \left( \frac{p}{d / \lambda} \right)^2 + \left( \frac{q d / \lambda - p x_{2,1} / \lambda}{ (d / \lambda) (y_{2,1} / \lambda)} \right)^2 = 3 p^2 + (2q - p)^2 \ge 4 .
\end{equation*}
Equality is achieved only when $p=1$, $q=0$ or $p=1$, $q=1$. The corresponding swept beam patterns are illustrated in Fig.~\ref{Fig:PatternA}, where it can be observed that no grating lobes appear across the entire steering range. This example further confirms that the classical half-wavelength limit is no longer a constraint in the context of dual linear array topologies.

\begin{figure}[!htbp]
  \centering
  \subfigure[]{
  \label{Fig:PatternA:Mesh}
  \includegraphics[width=0.48\columnwidth]{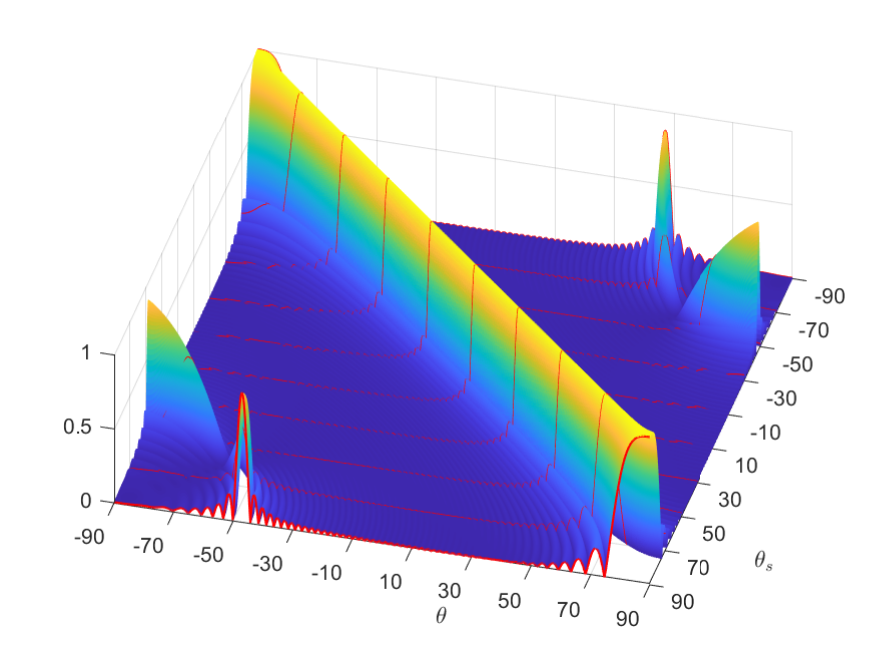}}
  \subfigure[]{
  \label{Fig:PatternA:Plot}
  \includegraphics[width=0.48\columnwidth]{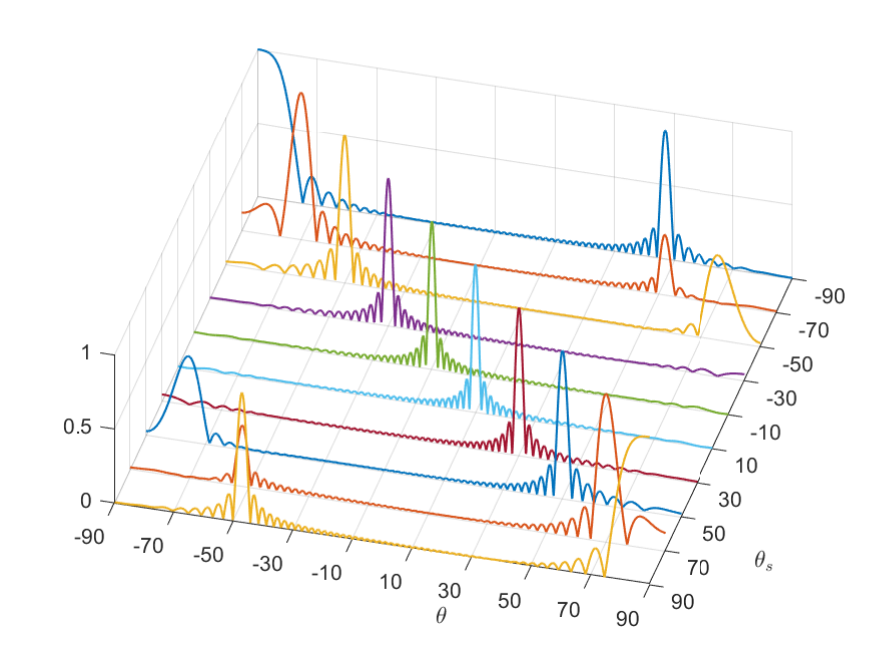}}
  \caption{Beam patterns for a dual linear topology with leading elements arranged in a regular triangular configuration, where $d = \sqrt{3} \lambda / {3}$, $x_{2,1} = \sqrt{3} \lambda / {6}$ and $y_{2,1} = 0.5 \lambda$. No grating lobes are observed across the entire range.}
  \label{Fig:PatternA}
\end{figure}

Finally, we point out that Theorems~\ref{T:theorem1} and \ref{T:SufficientGratingLobes} do not impose any constraint on the range of $\mathfrak{M} (\cdot)$. In practice, however, the angular field of view is often limited to the interval $[-\pi/2, \pi/2]$. This implies that, for certain topologies, even if (C3) is not satisfied, grating lobes may still not be observed within the region of interest $[-\pi/2, \pi/2]$. A detailed investigation of this phenomenon is beyond the scope of this paper; instead, we illustrate it through an example.

If the regular triangular topology is preserved but the side length is increased to $d = 0.6 \lambda > \sqrt{3} \lambda / {3}$ (so that both $x_{2,1}$ and $y_{2,1}$ scale proportionally), then (C3) is no longer satisfied for $p=1$, $q=0$ and $p=1$, $q=1$. However, numerical results in Fig.~\ref{Fig:PatternB} indicate that no grating lobes appear within the region of interest $[-\pi/2, \pi/2]$.

\begin{figure}[!htbp]
  \centering
  \subfigure[]{
  \label{Fig:PatternB:Mesh}
  \includegraphics[width=0.48\columnwidth]{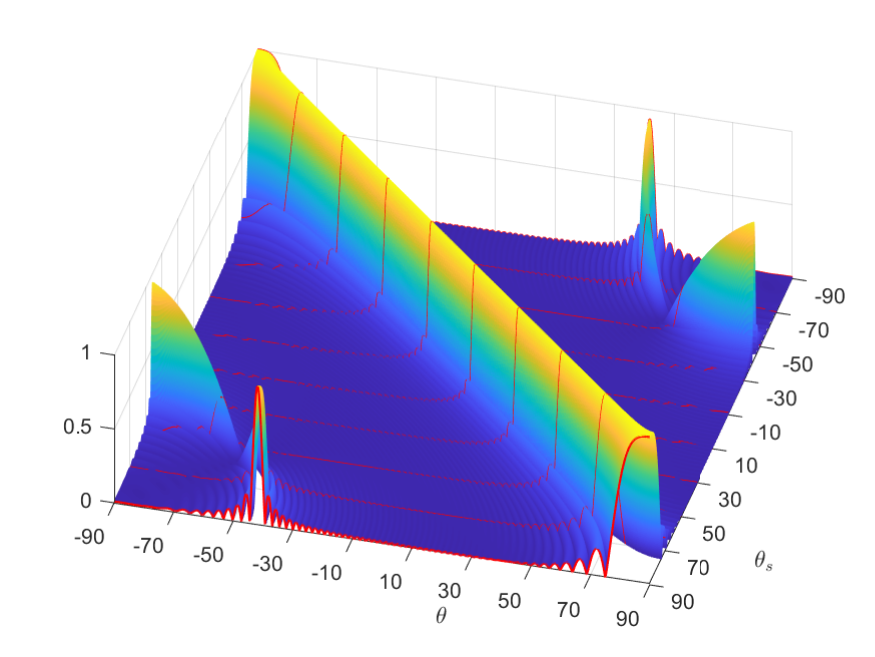}}
  \subfigure[]{
  \label{Fig:PatternB:Plot}
  \includegraphics[width=0.48\columnwidth]{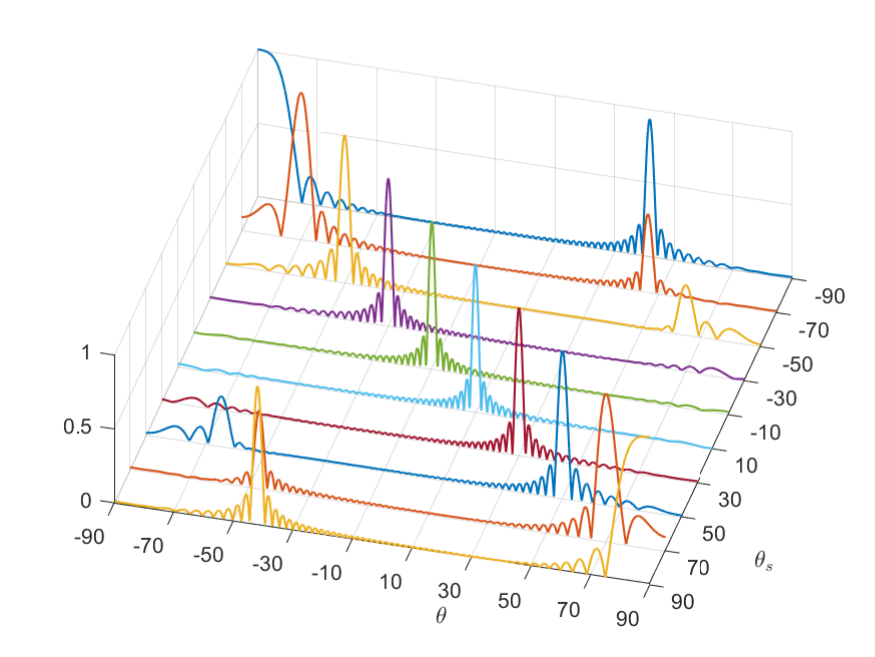}}
  \caption{Beam patterns for a dual linear topology with leading elements arranged in a regular triangular configuration, where $d = 0.6\lambda$. No grating lobes are observed within the region of interest $[-\pi/2, \pi/2]$.}
  \label{Fig:PatternB}
\end{figure}

\section{Fluctuations Induced by Spatial Perturbations}\label{S:IV}

In swarm antenna arrays, a challenge that cannot be overlooked is the perturbation of antenna element positions, primarily due to environmental disturbances or mechanical vibrations. These deviations degrade the coherence of the signal superposition and introduce random fluctuations into the radiation pattern. Understanding the effects of such perturbations is crucial for maintaining the performance of swarm-based antenna systems.

To model the array response under perturbations, we consider the following expression
\begin{equation}\label{E:PerturbedResponseTotal}
  \tilde{f}_{\mathbf{w}}(\theta) = \frac{1}{ \sum_{n=1}^{N} | w_{n} |} \sum_{n=1}^{N} w_{n} e^{j 2 \pi [ ( x_{n} + \Delta x_{n} ) \sin \theta + ( y_{n} + \Delta y_{n} ) \cos \theta ] / \lambda} .
\end{equation}
Here $\tilde{f}_{\mathbf{w}}(\theta)$ represents the distorted array response, and the perturbations $[\Delta x_n, \Delta y_n ]^T$ are assumed to follow a specified probability distribution. For analytical tractability, we adopt the phase-difference compensation weight design defined in \eqref{E:WeightSelection0}, which guarantees that the undisturbed response is coherently steered toward the desired direction $\theta_s$. Under this choice, the perturbed response at $\theta = \theta_s$ simplifies to
\begin{equation}\label{E:PerturbedResponse}
  \tilde{f}_{\mathbf{w}}(\theta_s) = \frac{1}{ \sum_{n=1}^{N} | w_{n} |} \sum_{n=1}^{N} | w_{n} | e^{j 2 \pi [ \Delta x_{n} \sin \theta_s + \Delta y_{n} \cos \theta_s ] / \lambda} .
\end{equation}
We now proceed to analyze the statistical properties of $\tilde{f}_{\mathbf{w}}(\theta_s)$.

\begin{thm}\label{T:SteerDirection}
Suppose the perturbations $[\Delta x_n, \Delta y_n ]^T$ are independently distributed according to a bivariate Gaussian distribution $\mathcal{N}(0, \Sigma_n)$. Then, the expected response at the steering angle $\theta_s$ is given by
\begin{equation}\label{E:MeanSteerDirection}
\begin{aligned}
\operatorname{E} \left[ \tilde{f}_{\mathbf{w}}(\theta_s) \right] = 
& \frac{1}{ \sum_{n=1}^{N} | w_{n} |} \sum_{n=1}^{N} | w_{n} | \\
& \exp \left( -\frac{2\pi^2}{\lambda^2} \begin{bmatrix} \sin\theta_s & \cos\theta_s \end{bmatrix} \Sigma_n \begin{bmatrix} \sin \theta_s \\ \cos \theta_s \end{bmatrix} \right).
\end{aligned}
\end{equation}
Moreover, the variance of the response at $\theta_s$ is given by
\begin{equation}\label{E:VarianceSteerDirection}
\begin{aligned}
\operatorname{Var} \left[ \tilde{f}_{\mathbf{w}}(\theta_s) \right] = 
  & \frac{1}{ \left( \sum_{n=1}^{N} | w_{n} | \right)^2} \sum_{n=1}^{N} | w_{n} |^2 \\
  & \left( 1 - \exp \left( -\frac{4 \pi^2}{\lambda^2} 
  \begin{bmatrix} \sin\theta_s & \cos\theta_s \end{bmatrix}
  \Sigma_n 
  \begin{bmatrix}
  \sin \theta_s \\ \cos \theta_s
  \end{bmatrix}
  \right) \right) .
\end{aligned}
\end{equation}
\end{thm}

\begin{proof}
We begin by evaluating the expectation of the complex exponential term $e^{j 2 \pi [ \Delta x_{n} \sin \theta_s + \Delta y_{n} \cos \theta_s ] / \lambda}$. This expression corresponds to the characteristic function of the random vector $[\Delta x_n, \Delta y_n ]^T$. Since $[\Delta x_n, \Delta y_n ]^T \sim \mathcal{N}(0, \Sigma_n)$, we obtain
\begin{multline}
\operatorname{E} \left[ e^{j 2 \pi [ \Delta x_{n} \sin \theta_s + \Delta y_{n} \cos \theta_s ] / \lambda} \right]  \\
= \exp \left( -\frac{2\pi^2}{\lambda^2} 
  \begin{bmatrix} \sin\theta_s & \cos\theta_s \end{bmatrix}
  \Sigma_n 
  \begin{bmatrix}
  \sin\theta_s \\ \cos\theta_s
  \end{bmatrix}
  \right).
\end{multline}
Substituting this result into the expression for $\tilde{f}_{\mathbf{w}}(\theta_s)$, as given in \eqref{E:PerturbedResponse}, yields the expected value in \eqref{E:MeanSteerDirection}.

Next, we compute the variance of \( \tilde{f}_{\mathbf{w}}(\theta_s) \). Since the perturbations across antenna elements are assumed to be independent, the total variance is given by
\begin{multline}
\operatorname{Var} \left[ \tilde{f}_{\mathbf{w}}(\theta_s) \right] 
= \frac{1}{\left( \sum_{n=1}^{N} |w_n| \right)^2} \sum_{n=1}^{N} |w_n|^2 \\
\operatorname{Var}  \left[ e^{ j 2 \pi [ \Delta x_n \sin \theta_s + \Delta y_n \cos \theta_s ]  / \lambda } \right].
\end{multline}
To evaluate the individual variances, we use the identity
\[
\begin{aligned}
& \operatorname{Var} \left[ e^{j 2 \pi [ \Delta x_{n} \sin \theta_s + \Delta y_{n} \cos \theta_s ] / \lambda} \right] \\
= & \operatorname{E} \left[ | e^{j 2 \pi [ \Delta x_{n} \sin \theta_s + \Delta y_{n} \cos \theta_s ] / \lambda} |^2 \right] \\
& - \left| \operatorname{E} [ e^{j 2 \pi [ \Delta x_{n} \sin \theta_s + \Delta y_{n} \cos \theta_s ] / \lambda} ] \right|^2 \\
= & 1 - 
\exp \left( -\frac{4 \pi^2}{\lambda^2} 
  \begin{bmatrix} \sin\theta_s & \cos\theta_s \end{bmatrix}
  \Sigma_n 
  \begin{bmatrix}
  \sin\theta_s \\ \cos\theta_s
  \end{bmatrix}
  \right).
\end{aligned}
\]
Substituting this result yields the variance expression in \eqref{E:VarianceSteerDirection}, completing the proof.
\end{proof}

We next investigate the statistical characteristics of the array response at non-steered directions. Unlike the coherent signal superposition observed at the steered direction, the response at non-steered angles exhibits a more complex pattern due to incoherent signal summation arising from phase misalignment among the elements.

To analyze this behavior, we consider the complex exponential term in \eqref{E:PerturbedResponseTotal}. When the perturbations are sufficiently small, a first-order Taylor expansion provides a valid approximation of the perturbed response. Applying this expansion with respect to $\Delta x$ and $\Delta y$ yields 
\begin{equation*}
\begin{aligned}
& e^{j 2 \pi [ ( x + \Delta x ) \sin \theta + (y + \Delta y) \cos \theta ] / \lambda} \\
\approx & e^{j 2 \pi [ x \sin \theta + y \cos \theta ] / \lambda} \left( 1 + j 2 \pi / \lambda \sin \theta \Delta x + j 2 \pi / \lambda \cos \theta \Delta y \right) .
\end{aligned}
\end{equation*}

Substituting the Taylor expansion into \eqref{E:PerturbedResponseTotal} yields the approximation $\tilde{f}_{\mathbf{w}}(\theta) \approx f_{\mathbf{w}}(\theta) + \Delta f_{\mathbf{w}}(\theta)$, as presented in \eqref{E:PerturbedBeamPattern}, where $f_{\mathbf{w}}(\theta)$ represents the ideal (unperturbed) array response, and $\Delta f_{\mathbf{w}}(\theta)$ captures the distortion induced by spatial perturbations. A key advantage of the formulation in \eqref{E:PerturbedBeamPattern} is its ability to decouple the nominal steering response from perturbation-induced distortion. This linearized model is particularly useful for analyzing the sensitivity of array performance to small positional deviations. The following theorem characterizes the statistical properties of $\Delta f_{\mathbf{w}}(\theta)$.

\begin{figure*}[htbp]
\begin{equation}\label{E:PerturbedBeamPattern}  
    \tilde{f}_{\mathbf{w}}(\theta) \approx f_{\mathbf{w}}(\theta)  + \underbrace{ \frac{j 2 \pi}{ \lambda \sum_{n=1}^{N} | w_{n} |} \sum_{n=1}^{N} \left(  w_n e^{j 2 \pi [ x_n \sin \theta + y_n \cos \theta ] / \lambda} \right) ( \sin \theta \Delta x_n + \cos \theta \Delta y_n ) }_{\Delta f_{\mathbf{w}}(\theta) } .
\end{equation}
\begin{equation}\label{E:Mean}
  \operatorname{E} [ \Delta f_{\mathbf{w}}(\theta) ] 
= \operatorname{E} \left[ \frac{j 2 \pi}{ \lambda \sum_{n=1}^{N} | w_{n} |} \sum_{n=1}^{N} \left(  w_n e^{j 2 \pi [ x_n \sin \theta + y_n \cos \theta ] / \lambda} \right) ( \sin \theta \Delta x_n + \cos \theta \Delta y_n ) \right] 
= 0 .
\end{equation}
\begin{equation}\label{E:Variance}
\begin{aligned}
\operatorname{Var} \left[ \Delta f_{\mathbf{w}}(\theta) \right]
& = \operatorname{Var} \left[ \frac{j 2 \pi}{ \lambda \sum_{n=1}^{N} | w_{n} |} \sum_{n=1}^{N} \left(  w_n e^{j 2 \pi [ x_n \sin \theta + y_n \cos \theta ] / \lambda} \right) ( \sin \theta \Delta x_n + \cos \theta \Delta y_n ) \right] \\
& = \frac{ (2 \pi)^2 }{ \lambda^2 } \frac{ \sum_{n=1}^{N} |w_n|^2 
\operatorname{Var} \left[ \sin \theta \Delta x_n + \cos \theta \Delta y_n \right] }{ \left( \sum_{n=1}^{N} |w_n| \right)^2} 
= \frac{ (2 \pi)^2 }{ \lambda^2 } \frac{ \sum_{n=1}^{N} |w_n|^2 \begin{bmatrix}\sin \theta & \cos \theta \end{bmatrix} \Sigma_n \begin{bmatrix}\sin \theta \\ \cos \theta \end{bmatrix} }{ \left( \sum_{n=1}^{N} |w_n| \right)^2} .
\end{aligned}
\end{equation}
\hrule
\end{figure*}

\begin{thm}\label{T:PerturbedPattern}
Suppose the perturbations $[\Delta x_n, \Delta y_n ]^T$ are sufficiently small and independently follow a bivariate Gaussian distribution $\mathcal{N}(0, \Sigma_n)$. Then, for $\theta \neq \theta_s$, the perturbation-induced fluctuation $\Delta f_{\mathbf{w}}(\theta)$ is approximately Gaussian with zero mean. Its variance is given by
\begin{equation}\label{E:Variance2}
\operatorname{Var} \left[ \Delta f_{\mathbf{w}}(\theta) \right] 
= \frac{ (2 \pi)^2 }{ \lambda^2 } \frac{ \sum_{n=1}^{N} |w_n|^2 \begin{bmatrix}\sin \theta & \cos \theta \end{bmatrix} \Sigma_n \begin{bmatrix}\sin \theta \\ \cos \theta \end{bmatrix} }{ \left( \sum_{n=1}^{N} |w_n| \right)^2} .
\end{equation}
\end{thm}

\begin{proof}
For any given $\theta$, $\Delta f_{\mathbf{w}}(\theta)$ is a linear combination of Gaussian random variables, as shown in \eqref{E:PerturbedBeamPattern}. By the closure property of Gaussian distributions under linear transformations, it follows that $\Delta f_{\mathbf{w}}(\theta)$ is itself Gaussian. The zero-mean property and the variance expression follow directly from the derivations in \eqref{E:Mean} and \eqref{E:Variance}, respectively. 
\end{proof}

To gain further insight into Theorems~\ref{T:SteerDirection} and \ref{T:PerturbedPattern}, we consider a specific scenario in which the phase-difference compensation strategy in \eqref{E:WeightSelection1} is applied, i.e., $|w_n| = 1$, and the positional perturbations in the $x$- and $y$-directions are assumed to be independent.

\begin{cor}\label{T:SteerDirection2}
Suppose the perturbations $[\Delta x_n, \Delta y_n ]^T$ are sufficiently small and independently drawn from an isotropic bivariate Gaussian distribution $\mathcal{N}(0, \sigma^2 I)$. If the phase-difference compensation strategy defined in \eqref{E:WeightSelection1} is employed, then the expected array response at the steering angle $\theta_s$ is given by
\begin{equation}\label{E:MeanSteerDirection2}
\begin{aligned}
\operatorname{E} \left[ \tilde{f}_{\mathbf{w}}(\theta_s) \right] = \exp \left( - 2 \pi^2 \sigma^2 /\lambda^2 \right).
\end{aligned}
\end{equation}
Moreover, the variance of the response at $\theta_s$ is
\begin{equation}\label{E:VarianceSteerDirection2}
\begin{aligned}
\operatorname{Var} \left[ \tilde{f}_{\mathbf{w}}(\theta_s) \right] = 
  & \frac{1}{ N } \left( 1 - \exp \left( - 4 \pi^2 \sigma^2 /\lambda^2 \right) \right) ,
\end{aligned}
\end{equation}
where $N$ denotes the total number of antenna elements. For directions $\theta \neq \theta_s$, the perturbation-induced fluctuation $\Delta f_{\mathbf{w}}(\theta)$ is approximately Gaussian with variance
\begin{equation}\label{E:Variance3}
\operatorname{Var} \left[ \Delta f_{\mathbf{w}}(\theta) \right] = \frac{(2 \pi)^2}{N} \left( \frac{ \sigma }{ \lambda } \right)^2 .
\end{equation}
Furthermore, the tail probability satisfies the bound
\begin{equation}\label{E:TailBound}
\operatorname{P} \left( | \Delta f_{\mathbf{w}}(\theta) | \ge t \right) \le 2 \exp \left( -\frac{ t^2 N }{ 2 (2 \pi)^2 ( \sigma / \lambda )^2 } \right) .
\end{equation}
\end{cor}

\begin{proof}
  The corollary follows directly from the assumptions $|w_n| = 1$ and $\Sigma_n = \sigma^2 I$.
\end{proof}

\begin{rem}
The significance of Theorems~\ref{T:SteerDirection} and \ref{T:PerturbedPattern}, along with Corollary~\ref{T:SteerDirection2}, lies in their explicit characterization of how perturbation-induced fluctuations depend on key system parameters, including the number of elements $N$, the wavelength $\lambda$, the perturbation strength $\sigma^2$, and the choice of weight design. Notably, \eqref{E:MeanSteerDirection} and \eqref{E:MeanSteerDirection2} demonstrate that spatial perturbations lead to a measurable degradation in the array response at the steered direction. Importantly, this degradation is governed solely by the perturbation strength and cannot be mitigated by increasing the number of antenna elements. The primary benefit of scaling up the array lies in reducing the variance of perturbation-induced fluctuations, as quantified by \eqref{E:VarianceSteerDirection2} and \eqref{E:Variance3}. In other words, as the array grows larger, the collective behavior becomes increasingly stable and predictable.
\end{rem}

To validate performance under positional perturbations, we apply independent Gaussian noise with a standard deviation of $0.1\lambda$ along both the x- and y-directions to the regular triangular configuration shown in Fig.~\ref{Fig:GreaterHalfWavelength:0_5774}. The corresponding beam patterns are presented in Fig.~\ref{Fig:PatternPerturb}, where fluctuations are observed across the entire steering range, although the overall pattern remains largely preserved.

\begin{figure}[!htbp]
  \centering
  \subfigure[]{
  \label{Fig:PatternPerturb:Mesh}
  \includegraphics[width=0.48\columnwidth]{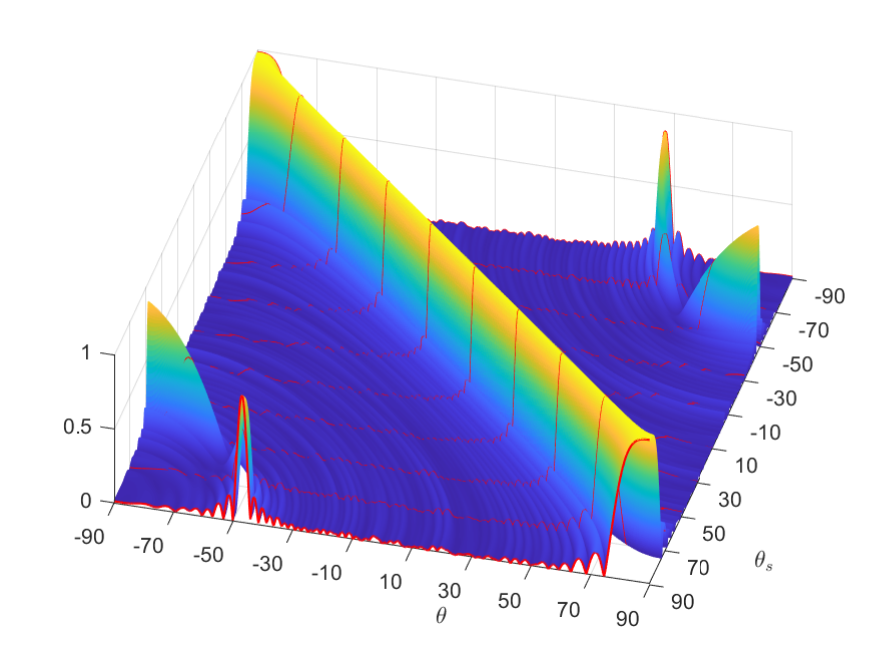}}
  \subfigure[]{
  \label{Fig:PatternPerturb:Plot}
  \includegraphics[width=0.48\columnwidth]{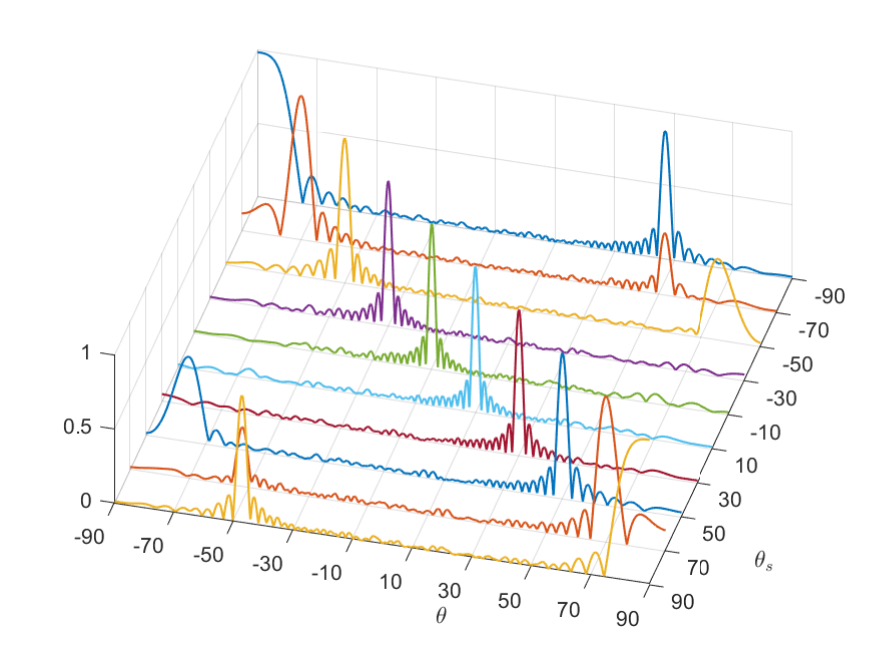}}
  \caption{Beam patterns for a dual linear topology with leading elements arranged in a regular triangular configuration, where $d = \sqrt{3} \lambda / {3}$. Independent Gaussian noise with a standard deviation of $0.1\lambda$ is applied along both the x- and y-directions. Fluctuations are observed across the entire steering range, although the overall pattern remains largely preserved.}
  \label{Fig:PatternPerturb}
\end{figure}

We validate the impact of the total number of antenna elements $N$ on the fluctuation of the array response. As indicated by \eqref{E:VarianceSteerDirection2} and \eqref{E:Variance3}, the variance of perturbation-induced fluctuations decreases inversely with $N$. This result implies that larger arrays yield more stable performance, even when the perturbation level at each antenna remains constant. To support this conclusion, numerical experiments are conducted, in which the absolute fluctuation $|f_{\mathbf{w}}(\theta) - \tilde{f}_{\mathbf{w}}(\theta)|$ is recorded over 500 trials. The averaged results are presented in Fig.~\ref{F:Unperturbed_Perturbed_Error_500} clearly demonstrating that increasing the number of elements suppresses random fluctuations at non-steered directions. Notably, at the steered direction $\theta_s = 0^\circ$, the absolute fluctuation remains at a similar level, confirming that the degradation at the steering angle cannot be mitigated by increasing the number of antenna elements, as established in \eqref{E:MeanSteerDirection2}.

\begin{figure}[htbp]
  \centering
  \includegraphics[width=0.8\columnwidth]{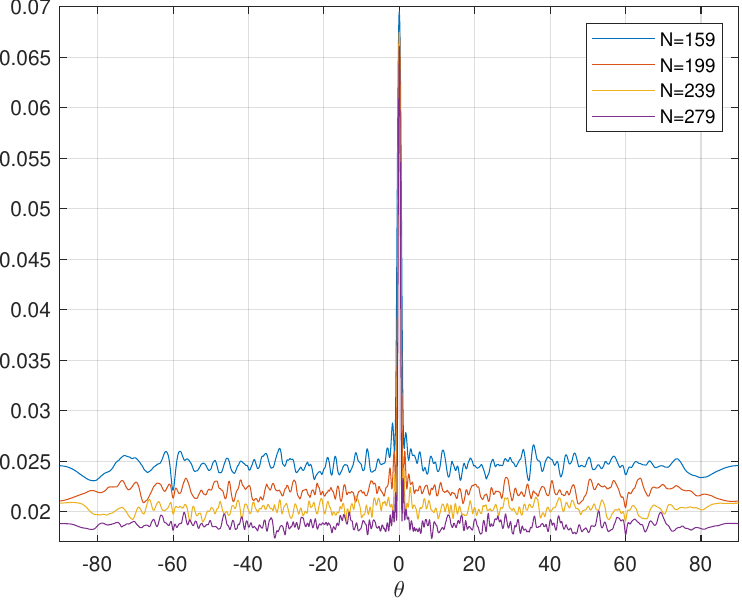}
  \caption{Averaged absolute fluctuation $|f_{\mathbf{w}}(\theta) - \tilde{f}_{\mathbf{w}}(\theta)|$ over 500 trials.}
  \label{F:Unperturbed_Perturbed_Error_500}
\end{figure}

\section{Emergent Deterministic Behavior in Disordered Antenna Arrays}\label{S:V}

In this section, we investigate the deterministic behavior inherent in the Euclidean random matrix associated with disordered swarm antenna arrays---a property that, although not widely explored, could prove valuable for communication and sensing tasks. In many practical scenarios, the positions of individual antenna elements are not fixed and may be randomly distributed within a specified range. In such cases, certain aggregate measurements exhibit statistical regularities that asymptotically converge to deterministic patterns. This phenomenon reflects a fundamental principle of large-scale systems, where microscopic randomness gives rise to macroscopic order.

\begin{figure}[htbp]
  \centering
  \includegraphics[width=0.7\columnwidth]{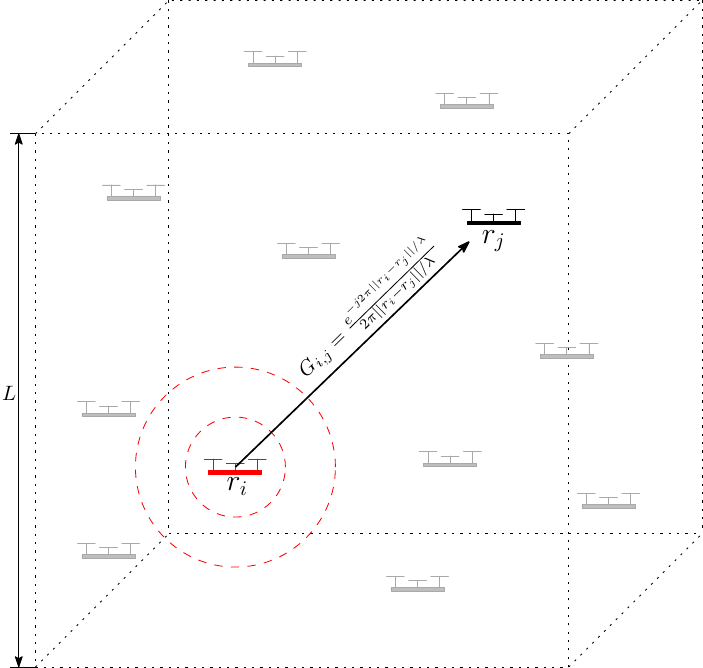}
  \caption{The positions of individual antenna elements are not fixed on regular grids but are randomly distributed within the cube.}
  \label{F:Cube}
\end{figure}

\begin{defn}[Euclidean Random Matrices]
Let $\{ \mathbf{r}_n = [x_n, y_n, z_n]^T \}_{n=1}^N$ denote the coordinates of $N$ antenna elements randomly distributed in a three-dimensional region of volume $V$, typically according to a uniform distribution. The Euclidean random matrix $G_N \in \mathbb{C}^{N \times N}$ is defined element-wise as
\[
G_N (i,j) = \frac{ \exp (-j 2 \pi \Vert \mathbf{r}_i - \mathbf{r}_j \Vert / \lambda ) }{ - 2 \pi \Vert \mathbf{r}_i - \mathbf{r}_j \Vert / \lambda } ,\quad \text{for} \ i\neq j,
\]
and
\[
  G_{i,i} = 0.
\]
\end{defn}

The Euclidean random matrix is closely related to the free-space Green's function, differing only by a scaling factor. Indeed, any Euclidean random matrix can be decomposed into its real and imaginary parts. For $i \neq j$, we write
\begin{equation*}
\begin{aligned}
G_N (i,j) = & C_N (i,j) + j S_N (i,j) \\
= & \frac{ \cos (-2 \pi \Vert \mathbf{r}_i - \mathbf{r}_j \Vert / \lambda ) }{ - 2 \pi \Vert \mathbf{r}_i - \mathbf{r}_j \Vert / \lambda } + j \frac{ \sin (-2 \pi \Vert \mathbf{r}_i - \mathbf{r}_j \Vert / \lambda ) }{ - 2 \pi \Vert \mathbf{r}_i - \mathbf{r}_j \Vert / \lambda }
\end{aligned}
\end{equation*}
A key theoretical result is that, under suitable statistical assumptions, the empirical spectral distributions of both $C_N$ and $S_N$ converge to deterministic limiting distributions as $N \to \infty$. 

\begin{defn}[Empirical Spectral Distribution]
The empirical spectral distribution of a Hermitian matrix $M_N \in \mathbb{C}^{N \times N}$ is defined as
\[
\mu_{M_N}(x) \equiv \frac{1}{N} \sum_{i=1}^{N} 1_{ \lambda_i(M_N) \le x }
\]
where $\lambda_1(M_N) \le \cdots \le \lambda_N(M_N)$ are the (necessarily real) eigenvalues of $M_N$, each counted according to its algebraic multiplicity.
\end{defn}

In the subsequent proposition, we let $L$ denote the side length of a cubic region, and define the antenna density as $\rho = N / L^3$.

\begin{prop}
Let $S_N$ denote the imaginary part of the Euclidean random matrix $G_N$, and define the parameter $\beta = 2.8N / (2 \pi L / \lambda)^2$. If $\beta < 1$, then the empirical spectral distribution of $S_N$ converges, as $N \to \infty$, to a deterministic distribution given by the Mar\v{c}enko-Pastur law with density
\[
d \mu (x) = \frac{1}{2 \pi \beta x} \sqrt{(x-a)^{+}(b-x)^{+}}
\]
where $a = (1 - \sqrt{\beta})^2$, and $b = (1 + \sqrt{\beta})^2$, $a^{+} = \max (0, a)$. For $\beta > 1$, the Mar\v{c}enko-Pastur distribution no longer accurately describes the spectral behavior of $S_N$.

On the other hand, for the real part $C_N$ of $G_N$, in the low-density limit where $\rho \lambda^3 \ll 1$ and $\beta \ll 1$, the empirical spectral distribution of $C_N$ converges to the Wigner semi-circle law with density
\[
d \mu (x) = \frac{ \sqrt{4 \beta - x }  }{ 2 \pi \beta} .
\]
In the opposite regime, where $\rho \lambda^3 \ll 1$ and $\beta \gg 1$, the empirical spectral distribution of $C_N$ converges to the Cauchy distribution with density
\[
  d \mu (x) = \frac{ 1 } { \pi (1 + x^2) } .
\]
\end{prop}

The proof of this proposition is mathematically intricate and falls beyond the scope of this paper; detailed derivations can be found in \cite{skipetrov2011eigenvalue}. To illustrate the deterministic behavior exhibited by Euclidean random matrices in swarm antenna arrays, we conduct a numerical experiment under the following configuration: the carrier frequency is set to $f = 1 \ \text{GHz}$, corresponding to a wavelength $\lambda = 0.3 \ \text{m}$, and the number of antenna elements is fixed at $N = 8000$. The antennas are uniformly distributed within a cubic region with a side length of either $L=20 \ \text{m}$ or $L=40 \ \text{m}$, corresponding to different density and $\beta$ regimes. The empirical spectral distributions of the real and imaginary parts of the Euclidean random matrix are computed and compared with the corresponding limiting laws. The results, plotted in Figs.~\ref{Fig:EuclideanRMT} and \ref{Fig:EuclideanRMT2}, reveal excellent agreement with the Wigner semi-circle law and the Mar\v{c}enko-Pastur distribution under their respective regimes.

\begin{figure}[!htbp]
  \centering
  \subfigure[]{
  \label{Fig:EuclideanRMT:1}
  \includegraphics[width=0.48\columnwidth]{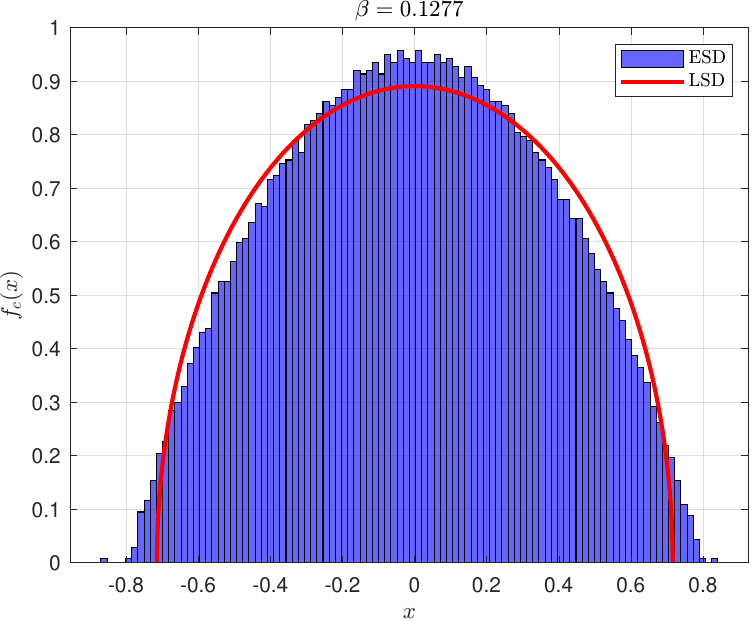}}
  \subfigure[]{
  \label{Fig:EuclideanRMT:2}
  \includegraphics[width=0.48\columnwidth]{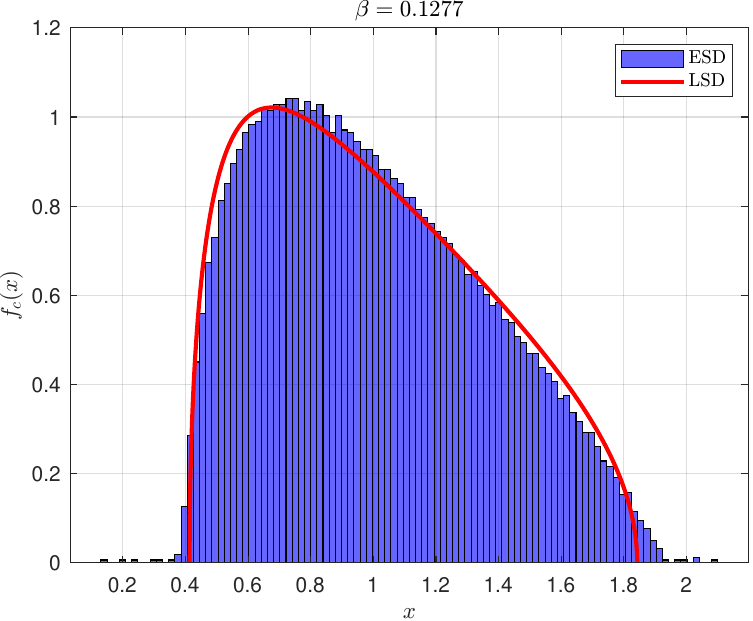}}
  \caption{Empirical and limiting spectral distributions of the real and imaginary parts of the Euclidean random matrix with $\lambda = 0.3 \ \text{m}$, $L=20 \ \text{m}$ and $N = 8000$, corresponding to $\beta = 0.1277$ and $\rho \lambda^3 = 0.027 \ll 1$. (a) Real part. (b) Imaginary part.}
  \label{Fig:EuclideanRMT}
\end{figure}

\begin{figure}[!htbp]
  \centering
  \subfigure[]{
  \label{Fig:EuclideanRMT2:3}
  \includegraphics[width=0.48\columnwidth]{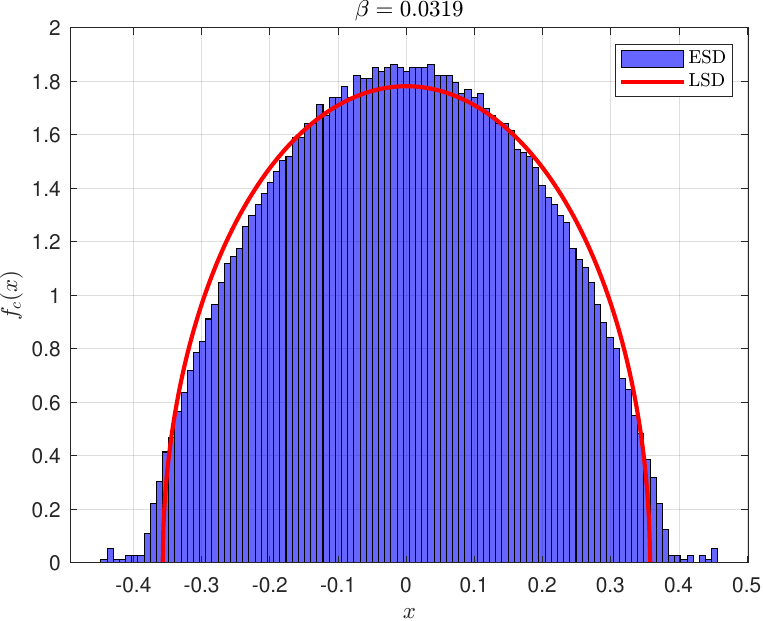}}
  \subfigure[]{
  \label{Fig:EuclideanRMT2:4}
  \includegraphics[width=0.48\columnwidth]{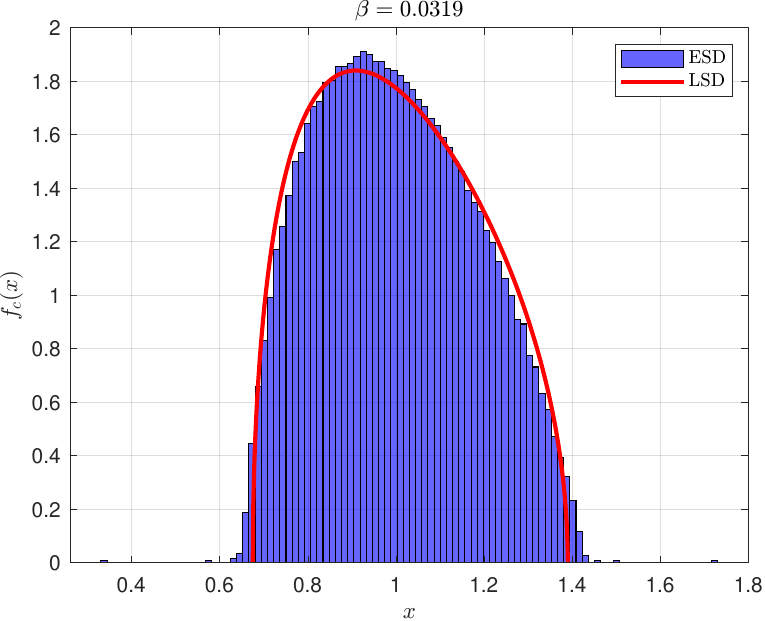}}
  \caption{Empirical and limiting spectral distributions of the real and imaginary parts of the Euclidean random matrix with $\lambda = 0.3 \ \text{m}$, $L=40 \ \text{m}$ and $N = 8000$, corresponding to $\beta = 0.0319$ and $\rho \lambda^3 = 0.0034 \ll 1$. (a) Real part. (b) Imaginary part.}
  \label{Fig:EuclideanRMT2}
\end{figure}

\section{Conclusion}\label{S:VI}

This work investigates the feasibility of coherent beamforming in swarm antenna arrays under both deterministic and stochastic conditions. We demonstrate that classical design constraints, such as half-wavelength spacing, can be relaxed without introducing grating lobes in multiple linear arrays. Our analysis shows that although spatial perturbations degrade the main lobe, coherent gain can be approximately maintained. Our analysis shows that while spatial perturbations degrade the main lobe, coherent gain can be approximately preserved. Furthermore, we identify emergent structural regularities in disordered arrays through the spectral convergence of Euclidean random matrices. These findings provide new theoretical foundations and practical design insights for enabling advanced functionalities in swarm antenna arrays.

\bibliographystyle{IEEEtran}
\bibliography{References}

\end{document}